\documentclass[12pt]{article}

\usepackage{amssymb,amsmath,amsfonts,eurosym,geometry,ulem,graphicx,caption,color,setspace,sectsty,comment,footmisc,caption,natbib,pdflscape,subfigure,array,hyperref,tikz}
\usepackage{natbib}
\usepackage{url}

\newtheorem{proposition}{Proposition}

\newenvironment{proof}[1][Proof]{\noindent\textbf{#1.} }{\ \rule{0.5em}{0.5em}}

\newcolumntype{L}[1]{>{\raggedright\let\newline\\arraybackslash\hspace{0pt}}m{#1}}
\newcolumntype{C}[1]{>{\centering\let\newline\\arraybackslash\hspace{0pt}}m{#1}}
\newcolumntype{R}[1]{>{\raggedleft\let\newline\\arraybackslash\hspace{0pt}}m{#1}}

\usepackage{natbib}

\title{Mechanism Design Approaches to Blockchain Consensus}

\author{Joshua S. Gans and Richard Holden}

\begin{document}

\maketitle

\begin{abstract}
   Blockchain consensus is a state whereby each node in a network agrees on the current state of the blockchain. Existing protocols achieve consensus via a contest or voting procedure to select one node as a dictator to propose new blocks. However, this procedure can still lead to potential attacks that make consensus harder to achieve or lead to coordination issues if multiple, competing chains (i.e., forks) are created with the potential that an untruthful fork might be selected. We explore the potential for mechanisms to be used to achieve consensus that are triggered when there is a dispute impeding consensus. Using the feature that nodes stake tokens in proof of stake (POS) protocols, we construct revelation mechanisms in which the unique (subgame perfect) equilibrium involves validating nodes propose truthful blocks using only the information that exists amongst all nodes. We construct operationally and computationally simple mechanisms under both Byzantine Fault Tolerance and a Longest Chain Rule, and discuss their robustness to attacks. Our perspective is that the use of simple mechanisms is an unexplored area of blockchain consensus and has the potential to mitigate known trade-offs and enhance scalability. \\
   
\end{abstract}

\vfill
\begin{footnotesize}
\noindent * Rotman School of Management, University of Toronto and NBER (Gans) and School of Economics, University of New South Wales (Holden). 
All correspondence to joshua.gans@utoronto.ca. The latest version of this paper is available at joshuagans.com. Thanks to Ethan Buchman and Scott Kominers for useful discussions and Raphael Mu for excellent research assistance. All errors remain our own. 
\end{footnotesize}
\newpage

\section{Introduction}

Blockchains rely on transaction messages being broadcast (to what is called a `mempool') where agents (`validators' or `miners') assemble them into a block that is appended to a public ledger. Those agents confirm that messages are valid (e.g., tokens transferred from a wallet actually are owned initially by that wallet) and the block is accepted as confirmed by other agents (i.e., that there is consensus regarding the validity of the proposed block). There are two challenges: (i) that consensus is reached and (ii) that consensus is over a set of truthful messages. For distributed ledgers, `truth' has a specific meaning: that the messages assembled into the block are those in the mempool without any being privately added by validators or intentionally excluded (or censored) by them.

Achieving distributed consensus is a challenge because of the dual goals of not having anything bad happen (\textit{safety}) and having something good happen (\textit{liveness}). Typically, improving the probability of one goal being achieved happens at the expense of reducing the probability that the other is achieved. Blockchains divide changes into blocks so that things can happen at a regular pace while they chain those blocks together through cryptographic hashing which assists in achieving safety as block manipulation requires the manipulation of other blocks. While this assists in achieving computing safety, practical reliance on block commits often builds in other buffers to further limit the ability to manipulate blocks in a way that impacts on outcomes outside of the ledger environment (e.g., double spending tokens). Thus, if methods can be found that improve reliance on the ledger, this improves \textit{effective safety} achieved. 

There are two broad methods of achieving consensus in blockchains. In each case, a \textit{serial dictatorship} model is used, whereby a block proposer is randomly selected from the pool of validators whose task it is to propose a block of transactions.\footnote{See \cite{gans2021consensus} for a review. } The first method, invented by \cite{nakamoto2008bitcoin}, is the \textit{longest chain rule} (LCR). This rule asks (but, importantly, does not require) validators to append blocks to the longest chain. So long as there is always a longest chain, this acts as a coordinating device.\footnote{For an analysis of potential coordination issues with LCR blockchains see \cite{biais2019blockchain} and \cite{barrera2018blockchain} for proof of work and \cite{saleh2021blockchain} for proof of stake blockchains.}  The second method relies on \textit{Byzantine Fault Tolerance} (BFT) whereby a proposed block is considered confirmed if at least two thirds of validators have sent a message `agreeing' to the proposed block (\cite{pease1980reaching}, \cite{buchman2016tendermint}, \cite{buterin2017casper}). Future proposed blocks must then be chained to the last confirmed block.\footnote{\cite{halaburda2021economic} show that when nodes are not presumed to be honest (or non-faulty) and consider their own payoffs, coordination problems can still arise.} It is generally the case that LCR blockchains achieve coordination in expectation faster (i.e., are more likely to satisfy liveness conditions) than BFT blockchains. 

While these methods achieve consensus, it is safe to say that consensus on the truth is left to be determined by crowd behaviour. For instance, in LCR blockchains, there can be forks where two equally long chains exist with the potential that at least one of them has untruthful blocks -- e.g., past blocks that exclude messages in order to facilitate a double spending attack. In BFT blockchains, attacks that subvert the operation of the blockchain by allowing consensus to be delayed are possible. In each case, so long as the majority or super-majority of participants (weighted by computational power in proof or work or token holdings in proof of stake) are engaging in truthful messaging -- that is, messages that reflect what has been broadcast to the mempool -- then truthful consensus can be achieved in equilibrium. Nonetheless, in trying to achieve coordination in this way, there is the potential for some disruption if adverse agents coordinated interventions. 

The question we address in this paper is whether there are more efficient and more reliable ways to achieve truth in consensus by designing and encoding mechanisms. Mechanism design is the branch of economics that deals with creating incentives for self-interested agents with information not known to the designer to reveal that information truthfully and still be willing to participate in the relevant economic activity. For example, auctions can be viewed as mechanisms for truthful revelation of the willingnesses to pay of buyers to a seller. Without a mechanism, buyers will not want to reveal their true value to a seller who might take advantage of that by charging a higher price. Similarly, without a mechanism, a seller has to guess buyer valuation when setting a price lest buyers choose not to buy or delay purchases. An auction, by specifying how a bid (a reflection of a buyer's true valuation) can be used, can be designed in such a way that the buyer has an incentive to tell the truth. This benefits buyers and the seller compared to a mechanism-less alternative but, critically, it relies on the mechanism being followed by the seller (\cite{akbarpour2020credible}) and bid information being communicated accurately (\cite{hurwicz1972informationally}, \cite{eliaz2002fault}).

Typically, mechanisms are conceived of as being designed and then implemented in a centralized manner. This would, on the face of it, put it at odds with being used in permissionless blockchains whose modus operandi is decentralized operation. However, while permissionless blockchains are often characterised as decentralized -- i.e., no one entity controls them and there is no one single point of failure -- a key part is centralized by virtue of there needing to be consensus on the state of the blockchain. Moreover, the blockchain's code is public and itself regarded as immutable. Thus, the ingredients for both a theoretical (i.e., to obtain truthful revelation using incentives) and practical use (i.e., transparent, unchangeable and unique code) are present for the use of mechanisms. This motivates our current examination of that possibility.

Our focus is on proof of stake protocols. In proof of work, the scope to deploy a mechanism does not exist because winning the computational game gives the proposer an unfettered right as to the block proposed. In proof of stake, by contrast, validating nodes are required to have committed an amount of tokens prior to potentially being selected to propose a node. That stake (technically, a bond) can then be used to provide incentives for any mechanism -- for instance, creating risk that if the node does not propose a truthful block that node will be worse off than had not participated in the mechanism at all. 

Indeed, for behaviors that can be readily identified as illegitimate -- such as proposing two conflicting blocks, being unavailable after promising to be available or proposing a block that isn't chained to the genesis block (or a checkpoint) -- proof of stake blockchains automatically fine misbehaving nodes in a process that is called `slashing' (\cite{buterinminimal}). This relies on the protocol designers identifying specific behaviours that may be associated with ill intent rather than something more straightforward and robust. Here our goal is to examine whether messages exist to ensure that validating nodes propose truthful blocks more generally using the information that exists amongst all nodes.

We construct an explicit mechanism under BFT and show that the unique equilibrium in this mechanism involves consensus on truthful blocks from any pair of nodes. In fact, this can be done with an arbitrarily small fine which, itself, does not occur on the equilibrium path. Moreover, under our mechanism there is no need for multiple rounds of confirmation to confirm a block--any randomly-selected pair of nodes suffices. Thus, it arguably has an efficiency property that improves the liveliness of these blockchains.

We also provide a bound for the share of the network that needs to be held by an attacker to succeed in a multi-node attack. A key insight here is that requiring nodes to send their messages before knowing the identity of the proposer strengthens this bound by an order of magnitude. In general, we conclude that our consensus mechanism is no more vulnerable to attacks than existing POW or POS protocols, and that the number of rounds the mechanism runs to confirm a block provides a kind of ``advance warning'' not provided by existing protocols.

For LCR blockchains, we offer an even simpler mechanism and show that in the unique (subgame perfect) equilibrium no dishonest forks arise. This stems from the fact that, under our mechanism, if a dishonest node is selected, they cannot get a transaction removed. This makes it suboptimal to dispute the transaction and the transaction is written to the blockchain. Given that, it is not worthwhile attempting the attack in the first place.

Our results rely on two features of the blockchain environment. First, because of cryptographic requirements associated with certain messages (such as a message sending tokens to another address), the message space is limited in certain important directions that limit the type of non-truths that can arise. Second, because of the public nature of parent blocks and the mempool, participating nodes know what the truth is. This second feature means that multi-stage mechanisms of the type examined by \cite{moore1988subgame} exist to ensure truthful revelation in equilibrium. Our task is to find mechanisms for specific blockchain environments. 

It is important to note that our purpose here is to show the potential benefits of mechanisms to aid in achieving truthful consensus on blockchains. Of necessity, we consider simplified environments and, thus, abstract away from practical difficulties associated with coding such mechanisms. However, we do believe that the broad framework we offer could be used as the foundation for practical implementation of mechanisms on blockchains and improve their operation.

While there has been much discussion of the use of mechanism design to inform aspects of the blockchain such as smart contracts (e.g., \cite{buterin2019flexible}, \cite{holden2021can}, \cite{gans2019fine}) only scant attention has been paid to blockchain consensus. For instance, \cite{leshno2020bitcoin} use a mechanism design approach to examine how nodes are selected to process transactions and receive block rewards and derive an impossibility result. \cite{garratt2022impossibility} examine truth-telling in blockchains and find another impossibility result. \cite{Roughgarden_2021} examines how mechanism design can improve transaction fee efficiency on blockchains. Finally, \cite{halaburda2021economic} do not look at mechanisms but look at economic incentives in BFT protocols. Here our approach is to examine whether explicit mechanisms can be used to substitute from other consensus resolving solutions in blockchains.

The remainder of the paper is organized as follows. Section 2 discusses mechanisms for Byzantine Fault Tolerance, while Section 3 turns attention to POS blockchains relying on the Longest Chain Rule. Section 4 contains some brief concluding remarks. 

\section{Mechanism for Byzantine Fault Tolerance}

The first broad consensus mechanism is proof of stake under Byzantine Fault Tolerance (or BFT). In the absence of an economic mechanism achieving consensus under BFT uses a voting mechanism. These voting mechanisms have the following steps:
\begin{enumerate}
    \item Transactions are broadcast as messages to the mempool
    \item Nodes stake and commit to be part of the validating pool
    \item Nodes observe messages 
    \item One node is selected to propose a block
    \item Other nodes choose whether to confirm or reject proposed block
    \item If at least two-thirds of nodes confirm the block it is accepted otherwise it is rejected and another node is selected to propose a block and the process begins again
\end{enumerate}
Typically, in voting to confirm a block nodes check the technical validity of the proposed block and also whether other nodes are confirming the same block. Thus, communication is multi-lateral and network-wide in the process of achieving consensus. Here we consider whether a mechanism can replace the voting process and limit communication to just two randomly chosen nodes before appending a new block to the chain. 

\subsection{A Simultaneous Report Mechanism}

The mechanism we propose is a special case of the Simultaneous Report (SR) Mechanism analysed by \cite{chen2018getting}.\footnote{The SR mechanism is a simplification of the multi-stage mechanisms explored by \cite{moore1988subgame}.} The baseline idea is that messages are broadcast publicly by blockchain users to the network and participating nodes assemble them into blocks of a fixed size based on time broadcast. When a block is proposed to be committed to the blockchain, each node has in their possession a block of messages they have received. We assume that this block is common across all nodes, however, there are no restrictions on nodes in proposing an alternative block. The goal is to ensure that nodes, while able to propose alternative blocks, only propose and accept truthful blocks.

Suppose there are nodes, $i \in \{1,...,n-1\}$ each of whom assemble ledger entries into a block of fixed size. If a node ends up proposing a block that is accepted, they receive a block reward, $R$. There is also an $n$th node who proposes a manipulated block. If that block is accepted, they receive a payoff $\theta$ that is private information in addition to the block reward, $R$. Nodes can send any message from a countably infinite set.

Consider the following mechanism that is run after messages have been sent to the mempool: 
\begin{enumerate}
    \item One node is randomly chosen to be the \textit{proposer}, $p$, and another node, $c$, is chosen to be the confirmer.
    \item The proposer proposes a block in the form of a message, $M_p$ while the confirmer sends a message, $M_c$. 
    \item If $M_p = M_c$, then the block is committed to and added to the blockchain. The proposer receives $R$.
    \item If $M_p \neq M_c$, then the challenge stage begins with both $p$ and $c$ being fined, $F > 0$. 
\end{enumerate}

The \textbf{challenge stage} involves:
\begin{enumerate}
    \item $p$ sends a new message $M^C_p$ based on knowledge that there is a disagreement.
    \item If $M^C_p = M_c$, then $M^C_p$ is committed to the blockchain, $p$ receives $R$, and $c$ is refunded $F$.
    \item If $M^C_p \neq M_c$, then $p$'s proposal is discarded and the process begins again with $p$ and $c$ excluded from subsequent rounds.
\end{enumerate}
Given this, we can prove the following:
\begin{proposition}
Suppose the true block is $M_T$. Then the unique subgame perfect equilibrium outcome for the mechanism for any pair of nodes is $M_p=M_c=M_T$.
\end{proposition}

\begin{proof}
Suppose that the selected pair does not include node $n$. Then working backwards, if $M_p \neq M_c$, then $M_p \neq M_T$, $M_c \neq M_T$ or both as the message space is a (countably) infinite set. In this case, the challenge stage is initiated and $p$ has the opportunity to send a new message. If $M_p = M_T$, then there is zero probability that $p$ could send $M^C_p = M_c$ and so $p$'s proposal is discarded and both nodes receive $-F$. If $M_p \neq M_T$, then by selecting $M^C_p = M_T$, then with some probability (possibly equal to 1), $p$ receives $R-F$ rather than $-F$ with certainty by choosing some other message. Thus, in the challenge stage $M^C_p = M_T$. Anticipating this, it is optimal for $p$ to set $M_p = M_T$ and $c$ to set $M_c = M_T$.

Now suppose that the selected pair includes node $n$. If node $n$ is the confirmer and the challenge stage is reached, then, we have already shown that the proposer will set $M^C_p = M_T$. Given this, node $n$ will find it optimal to set $M_c = M_T$ and earn $0$ rather than $-F$.

Alternatively, if node $n$ is the proposer, by our earlier argument, the other node will set $M_c = M_T$. If $n$ sets $M_p \neq M_T$, then there is a challenge round. In that round, $n$ will earn $R-F$ by setting $M^C_p = M_T$ and $-F$ otherwise. Given this, it is optimal for $n$ to set $M_p = M_T$ as it will earn $R$ rather than $R-F$. 
\end{proof}
\\

\noindent It is easy to see that in the challenge stage, if $M_c$ is the truth, $p$ knows this and so finds it worthwhile to set $M^C_p = M_c$ and receive $R$. If $M_p$ is the truth, $p$ has a problem as it does not know what $M_c$ was. In this case, it ends up setting $M^C_p \neq M_c$ and receiving $0$. Thus, the truth is revealed regardless. There seems something harsh about this last step as the proposer may be reasonable and still punished. That would arise only if $c$ has an incentive to message something other than the truth. If they message the truth, then they get $0$ as they expect $p$ to revise their message at least and receive a refund of $F$. If they message something else, then $p$ will never be able to guess that and so they will lose $F$. Thus, $c$ has no incentive to do anything other than be straightforward.

It is useful to stress how remarkably powerful this mechanism is for obtaining consensus on truthful blocks. We note the following:
\begin{itemize}
    \item $F$ can be arbitrarily small and the true block will be confirmed by any pair. 
    \item Any randomly-selected pair is sufficient to confirm a block. Unlike BFT mechanisms, there is no need for multiple confirmation rounds, pre-commits or messages sent from more than two nodes. Once block transactions have been communicated publicly and formed into the truthful block, the mechanism can take place and confirmation is instantaneous.
    \item If there is more than one node who has a private value that arises should a block other than the truthful block occur, the true block will still be confirmed. This outcome occurs even if two nodes with private preferences happen to be paired. This will happen so long as those nodes have preferences for distinct blocks. However, as we explore in the next section, if nodes have preferences for the same non-true block, the game is more complex (e.g., if there is a coalition of nodes).
    \item A key part of the mechanism is that while all nodes have common knowledge of the message for the true block, privately preferred blocks are unknown beyond individual nodes. This subverts any method by which coordination could arise on a non-true block. Once again, a coalition of nodes with a non-true preferred block could potentially coordinate and subvert the mechanism.
\end{itemize}

\noindent It is worth emphasising that the mechanism does rely critically on the messages of the true block being perfect and common knowledge. If this was not the case, the proof would be more complicated but we conjecture that it will still hold given the results of \cite{chen2018getting} that show that SR mechanisms are robust to some informational imperfections.

Finally, it is useful to note some practicalities in terms of implementing this mechanism. The mechanism relies on the two nodes being selected randomly. As will be discussed in detail below, randomness plays an important role in the mechanism working when there is more than one attacking node. Finding a pure randomisation device on-chain is a challenge for blockchains and, thus, we expect that it is likely that this part of the mechanism will rely on an external randomisation input. The mechanism itself can itself run on-chain but it would be as a smart contract coded into the protocol. This is something that is a feature of many proof-of-stake protocols for other elements. 

\subsection{Robustness to Multi-Node Attacks}

As noted above, while the proposer, if selected, has an incentive to tell the truth this is based on a specific assumption that could not be guaranteed for a permissionless blockchain (and maybe not all permissioned ones either): that the proposer and confirmer are different entities. What if the proposer and confirmer are the same person or part of an attacking coalition that share the same incentive to confirm an alternative, non-true block? 

Suppose that an attacking coalition has a share, $s$, of all nodes.  If they are the proposer, then, with probability $s$ they will be able to confirm the distorted block and receive $R+\theta$ and receive $-F$ otherwise. If they do not distort, they receive $R$ with certainty. (This assumes that the attacker has full knowledge of the fact that they are the proposer before setting $M_c$). Thus, the proposer will try to attack if:
\begin{equation}
s(R+\theta)-(1-s)F > R  \implies s > \frac{R+F}{R+\theta+F} \label{attack1}
\end{equation}
From this, it can be seen that this mechanism is not robust to an attack if the attacker has sufficient share ($s$) of the network.

We can compare this threshold to that typically considered for BFT networks. Attacks that may delay the confirmation of transactions are not possible in those networks if $s < \frac{1}{3}$. Notice here that the SR mechanism lowers this possibility if $R+F > \frac{\theta}{2}$. Thus, depending on the environment, the consensus protocol proposed here may be more secure than the usual BFT protocol. Moreover, security can be enhanced by increasing $R+F$ rather than exogenously set as it is under usual BFT consensus.

If we require nodes to send their messages \textit{before} knowing who is a proposer, then this slightly changes the equation. Now the attacker only succeeds with probability $s^2$ but also has an additional cost in that their confirmer role automatically leads to a fine. Thus, an attacker's choice depends on:
\[s^2(R+\theta)-s(1-s)F-(1-s)F > R\]
and so the threshold becomes $s > \sqrt{\frac{R+F}{R+\theta+F}}$ which is an order of magnitude stronger (modulo, integer issues which we ignore here).\footnote{In this case, the SR mechanism is more secure than the usual BFT consensus if $R+F > \frac{\theta}{8}$.}

\subsection{Patient Multi-Node Attacks}

The above calculations assume that an attack is considered by the attacker to be a once-off opportunity. This certainly is the case if, in the next round, truthful consensus is reached and the opportunity for an attack is removed. However, when there are multiple nodes, the incentive compatibility conditions in our proposed mechanism need to be reformulated for the possibility that, should consensus not be reached with one pair, at least one additional round of the mechanism will result with a new pair of nodes. 

To consider this, note the following:
\begin{itemize}
    \item Truthful consensus will arise if both selected nodes are honest (i.e., have no distortion payoff, $\theta$). At the outset, this happens with probability $(1-s)^2$;
    \item Truthful consensus may arise if a proposer is honest and a confirmer is potentially not. Recall that in the mechanism, the proposer will propose an honest block and continue to do so in a challenge round. Thus, if the confirmer proposes anything other than the true block, they are fined $F$ and another pair is selected. Let $V_{sN-1,(1-s)N-1}$  be the expected payoff to the attacking coalition if there are $sN-1$ remaining coalition nodes and $(1-s)N-1$ remaining honest nodes. Then the incentive compatibility constraint for the confirmer is $V_{sN-1,(1-s)N-1} < F$. This scenario occurs, initially, with probability $(1-s)s$.
    \item What happens if the proposer is not honest? Recall that nodes cannot see each others messages in the mechanism -- only if they match or not. Thus, if a proposer of a distorted block is matched with an honest node, they will have a choice as to whether to adjust their message or not. In effect, they can either continue the disagreement or confirm the true block (under the assumption that the node they are paired with must be honest). Thus, their incentive compatibility constraint in the challenge stage is to propose an honest block if $V_{sN-1,(1-s)N-1} < 0$. If this constraint is satisfied, then a proposer will propose a distorted block only if $s(R+\theta) - (1-s)F \ge R$. If this constraint is not satisfied, then the proposer will propose a distorted block only if $s(R+\theta) - (1-s)(F-V_{sN-1,(1-s)N-1}) \ge R$.
\end{itemize}
Examining these conditions, therefore, requires solving for $V_{sN-1,(1-s)N-1}$.

Exploring $V_{sN-1,(1-s)N-1}$, note that:
\begin{multline}
    V_{sN-1,(1-s)N-1} = \tfrac{sN-1}{N-2}\Big( \max\{\tfrac{sN-1}{N-2}(R+\theta)-\tfrac{(1-s)N-1}{N-2}(F-V_{sN-2,(1-s)N-2}),R\} \\
    + \tfrac{(1-s)N-1}{N-2}\max\{V_{sN-2,(1-s)N-2}-F,0\}\Big)
\end{multline}
The first probability is the probability that a potential attacker is selected as a proposer or confirmer. If they are a proposer or confirmer, the outcome depends upon whether their incentive compatibility constraint is satisfied or not. These are the next two terms within the brackets. Note that each of these depend upon $V_{sN-2,(1-s)N-2}$ which is the attacking coalition's expected payoff should no consensus be reached in the round. 

An important property of the recursive structure here is that, for $s \le \frac{1}{2}$, $V_{sN-1,(1-s)N-1} \ge V_{sN-2,(1-s)N-2}$ and this property continues between rounds if no consensus is reached. In this case, if the incentive compatibility constraint is satisfied for the proposer in the first round, it will continue to be satisfied and, thus, there will be no incentive for the attacker to start or continue the attack so long as $s < \frac{R+F}{R+\theta+F}$.

What if $s > \frac{1}{2}$? In this case, the expected payoff from an attack could rise between rounds. Indeed, if the penultimate round is reached, then the attacker knows they will have both the proposer and confirmer with certainty in the final round allowing them to confirm the distorted block and receive a certain payoff of $R+\theta$. In this case, in the penultimate round, there is one honest node and, say, $x$ nodes of the coalition. The incentive compatibility condition for the penultimate round would have to be such that $\frac{x}{x+1}(R+\theta)-\frac{1}{x+1}(F-R-\theta) < R$ or $(x+1)\theta < F$. Note that this is equivalent to  $s < \frac{F}{\theta}-\frac{1}{2}$ as $s=\frac{1+2x}{2}$ or $x=\frac{2s-1}{2}$. It can be seen that so long as $F > \frac{3}{2}\theta$, then an attack is never worthwhile.

This makes it appear that $F$ has to be very large. In fact, the protocol could adjust $F$ so that it \textit{increases with each round}. Precisely what this formula would be, however, is a complex matter and we have not been able to calculate it yet. In any case, as $\theta$ is a free variable calibrating it to that would be a challenge.

That said, the conditions here are similar to those regarding so-called ``majority" attacks under both POW and POS. The difference is rather than building out forks of existing chains, here the attack conditions take place on a block by block basis. The effect is the same, as is the cost of an attack in terms of resources -- real or financial (see \cite{gans2021consensus}). Our contention would be that our proposed consensus mechanism is not more vulnerable than existing ones but that it has the advantages that (a) it is significantly more efficient to run and operate under normal conditions; (b) continues to make an attack probabilistically difficult including forked-chain attacks and (c) can be indicated by the number of rounds a mechanism needs to run to confirm a block -- that is, in the absence of an attack, there will be few rounds while an attack with a close to 51 percent majority will still likely take many rounds. We believe that it would be possible to calculate these probabilities more precisely but that is left for future work.

\section{Mechanism to Resolve Forks}

Unlike the BFT protocol, POS blockchains that rely on the LCR expect forks to arise.\footnote{While we don't explore this possibility in this paper, BFT blockchains do sometimes generate forks which require lengthy and costly human intervention and coordination to resolve. Our fork resolution mechanism could be alternatively used even in conjunction with BFT consensus; although we only explore its use in LCR type consensus here.} A fork is a situation where two chains of equal length are confirmed to the protocol that have a common ancestor but distinct blocks thereafter. When this occurs nodes must decide which one to append new blocks to so that a new longest chain arises which becomes the consensus chain. Because there is no voting protocol, LCR blockchains can have blocks confirmed very quickly but because of the possibility of forks, consensus may not be final. 

Forks will arise simply because there may be lags in the communication of confirmed blocks to chains and different sets of nodes may work on distinct blocks for a time. One procedure POS blockchain protocols use to limit this lack of coordination is to punish nodes who work on more than one chain. This can happen because nodes have the potential to earn rewards for confirming blocks regardless of the chain that arises and so do not necessarily have an incentive to select one over another.\footnote{\cite{saleh2021blockchain} argues that a proposing node does have an incentive as only the chain that survives as the longest chain will earn them a block reward while persistent lack of coordination devalues tokens and hence those rewards and stakes held.} Blockchains resort to slashing (that is, fining nodes some portion of their stake) to incentivise nodes to work on one chain.\footnote{In POW protocols working on more than one chain requires actual resource expenditures rather than simply staking digital tokens so this incentive issues does not arise.}

While forks may be considered a necessary inconvenience, because they can arise as a matter of course, there is potential for nodes to create forks for their own purposes. For instance, a node may create a blockchain fork in order to nullify transactions previous confirmed to past blocks. The reason they might do this is in order to allow a transaction that would otherwise be invalid to current or future blocks. The most famous example of this is the `double-spend' problem whereby an agent attempts to spend their own tokens twice by confirming a transaction to a past block and then attempting to have that transaction omitted so they can spend those tokens again. This type of intervention relies on some agents relying on the past block confirmation to trigger, say, the transfer of real world goods to the agent. Given this \cite{nakamoto2008bitcoin} recommended that agents not consider transactions economically final until a certain number of blocks had been confirmed to the chain following the block in question. However, the potential to revert or nullify past transactions reduces the reliability of the blockchain and the speed at which transactions and messages can be accepted as final.

It is instructive to explain in more detail how an attack -- such as a `double-spend' attack -- actually works on POS with the LCR rule. An attacker privately works on a chain that nullifies a transaction in an already confirmed block that is part of the main chain. The transaction is removed so that the tokens remain in the attackers account to be spend again. The goal is then to surface the private chain once it is as long as (or longer than) the current main chain; creating a fork. In proof of work, doing this is non-trivial as this requires, at least, expending more electricity than those nodes working on the main chain. In proof of stake, as only the tokens of the attacker are staked on the private chain, this is simpler. What is more, the attacker, should they succeed, also receives the block rewards, $R$, associated with the number of blocks, $k$, they have worked on privately. Thus, their total reward from succeeding in this attack is $kR+\theta$ (where $\theta$ is the benefit from the double spend).

The attack is made more difficult if there are robust penalties to other nodes from staking tokens to both chains in a fork. If other nodes have ``nothing at stake," they do not have an incentive to `pick sides' in the fork. If, eventually, nodes that stake tokens in both chains can be detected, they can be fined an amount that disincentivizes that behavior. However, detecting this can be difficult (\cite{bentov2016cryptocurrencies}, \cite{brown2019formal}). Another approach, to be employed in Ethereum's Casper upgrade, is to punish nodes who stake on the `wrong' chain.\footnote{The means of tracking these outcomes involves the use of code called a `dunkle'; see \url{https://eth.wiki/concepts/proof-of-stake-faqs0.}} This incentivizes nodes to do more work in establishing what they believe the right chain to be. However, this could be a significant challenge in that most nodes may not want to play an active role in resolving forks and would prefer to following the LCR. Thus, at best, current mechanisms deter nodes from staking on different chains and so make it harder for an attacker to create an alternative longest chain. However, there is no real mechanism to determine which chain is the `true' chain. That is where we believe a mechanism can help.

The potential for a mechanism that might resolve blockchain forks arises because it is difficult to tell if they are accidental or deliberate. In the latter case, they are akin to an ownership dispute. A natural question is whether there is a mechanism that can quickly determine what the ``true'' chain is? Adapting our analysis of Solomonic disputes (\cite{gansholden2022}), we answer this question in the affirmative.

A fork is only consequential if different nodes claim that said fork is the correct one. Thus, there are competing ownership claims about the ``true'' chain. For simplicity call the competing forks $A$ and $B$. The information structure is such that the nodes claiming the fork $A$ is the true one know they are honest. Nodes claiming that $B$ is the true chain know they are not honest.

\subsection{A Solomonic Mechanism}

Now consider the following mechanism. If a fork appears (without loss of generality, $B$) and is within $x$ of the same number of blocks as fork $A$ then the following mechanism is run between nodes that claim to hold a full record of the blockchain.

\begin{enumerate}
    \item For each fork, the blocks subsequent to the last common parent are unpacked and transactions are compared. 
    \item Valid messages that appear in both sets are immediately confirmed and at the minimum of the time stamp between the two forks.\footnote{Valid messages will have different time stamps but otherwise the same content. However, as they cover multiple blocks there are some practical issues to resolve in comparing messages as being equivalent on each chain of the fork.} Other messages are collected and marked as disputed. 
    \item One node from each chain is selected at random ($a$ for $A$ and $b$ for $B$). The node from the chain where a transaction does not appear is asked to confirm that the transaction is invalid. In this case, both are fined $F$, and they enter the dispute stage for each disputed transaction. 
\end{enumerate}

\noindent The \textbf{dispute stage} involves:
\begin{enumerate}
    \item If the transaction appears in $A$ and not in $B$, $a$ is asked to assert the legitimacy of the transaction. If $a$ asserts, then the transaction remains and the fine is burned. If $a$ does not assert, the transaction is discarded and $b$ has their fine refunded.
    \item If the transaction appears in $B$ and not in $A$, $b$ is asked to assert the legitimacy of the transaction. If $b$ asserts, then the transaction remains and the fine is burned. If $b$ does not assert, the transaction is discarded and $a$ has their fine refunded.
\end{enumerate}

\noindent Note that there may be numerous transactions that appear in one chain and not the other. The procedure here, between the two selected nodes, would be conducted for each disputed transaction with the roles assigned depending on which chain the disputed transaction appears. 

We now need to specify preferences of each type of node. An honest node has an interest in preserving the true blockchain. That preference has a monetary equivalent value of $H$ and they have a disutility arising from another blockchain being built upon of $D$ with $H > D$. By contrast, a dishonest node is only interested in having their preferred chain continue for which they receive a monetary equivalent value of $\theta$. 

Given this, we can now prove the following:

\begin{proposition}
Dishonest forks do not arise in any subgame perfect equilibrium where at least one honest node is selected.
\end{proposition}

\begin{proof}
Without loss in generality, suppose that $A$ is the true blockchain and a dishonest fork, $B$, arises with a transaction omitted from a past block. If $b$ confirms that the transaction is invalid, $a$ and $b$ are then selected and fined and the dispute stage begins. In the dispute stage, $a$ is asked to assert or not assert the legitimacy of the transaction. If $a$ asserts, as the transaction is valid and $a$ is honest, $a$ receives $H-F$ and $b$'s payoff remains at $-F$. If $a$ does not assert, $a$ receives $D-F$ and $b$'s payoff becomes $\theta$. In this case, $a$ has a preference to assert. 

In the first round, anticipating this, $b$ will decline to confirm the removal of the transaction as this will result in a fine of $F$ for sure and no transaction being removed. Given that, it is not worthwhile attempting the attack in the first place. 

Now suppose that there are no dishonest nodes and that, in this case, $b$ is a node on a chain where the transaction does not appear (e.g., it may have been missed because of network issues). In this case, as $b$ knows they are honest and that dishonest nodes only attempt to remove transactions and not place them, $b$ is choose to not assert the transaction is invalid. Thus, the transaction will remain and neither node will be fined. 
\end{proof}
\\

\noindent The intuition for this result is simple. The mechanism is designed based on the notion that (1) nodes have information as to which chain they regard as truthful and which is not; and (2) the attacker is trying to have a past transaction/message removed from the blockchain. The mechanism gives the opportunity for a node that had staked on each chain to be matched and to confirm whether there is a dispute over a transaction that appears in one but not in the other. If there is an agreement, the transaction is confirmed or removed as the case may be and the mechanism ends. If there is a dispute, however, both nodes are fined which creates an incentive to avoid the dispute. The dispute stage then focusses on the node from the chain that asserts the transaction is valid. As an attack involves a dishonest node trying to remove a transaction, this node is presumptively honest and so are given control over the decision. Thus, the transaction remains. Given this is the outcome of the dispute a dishonest node will not trigger the dispute (as they will be fined) nor create a dishonest fork because that is costly and will be unsuccessful. 

Importantly, the mechanism takes into account the possibility that chains have arise as a result of, say, network latency issues, rather than an attack. In this case, both nodes are honest. Given that nodes know that attacks involve the removal of transactions, when the mechanism is triggered by an accidental fork, the node from the chain without the transaction will agree their fork should be discarded if the transaction is, indeed, valid. Thus, the `correct' fork persists.

Nonetheless, as the proposition qualifies, the mechanism does not prevent dishonest forks from arising when neither selected node is honest. This would require an attacker to distribute nodes across both forked chains; a possibility we consider next.

\subsection{Multi-Node Attacks}

As noted earlier, the `traditional' strategy for a double-spend attack on POS with the LCR is for an attacker to create private chain without the transactions that appear in the main chain and then build that change to be at least as long, and perhaps longer, than the main chain. The idea of the Solomonic mechanism is that when a fork appears, regardless of the length of competing chains, the mechanism resolves the dispute. As Proposition 2 shows that, so long as one honest node is selected to be part of the dyad in the mechanism (which will happen if the attacker has concentrated its resources to build the private chain), the dishonest chain will be discarded. Hence, if it is expected that an honest node will be selected, there is no incentive for an attacker to create a dishonest fork in the first place. 

The Solomonic mechanism is engaged as a mechanism to resolve forks in lieu of a race to produce the longest chain. One implication of this is that there is no longer a concern about the `nothing-at-stake' problem. Recall this problem arose because there were incentives for nodes to stake on two chains that were part of a fork which would have the implication of slowing down the resolution of one chain into the longest chain. At the extreme, the nothing-at-stake problem could make it easier for an attacker to create the longest chain. Hence, current POS networks using the LCR, create penalties for nodes who stake on multiple networks. However, when a Solomonic mechanism is used to resolve which chain should be confirmed, whether nodes stake on multiple chains is no longer a concern and, therefore, there is no reason to discourage this, let alone impose a slash penalty. Indeed, as will be argued here, there is reason to encourage multiple chain staking.

The reason for this is that an attacker has an interest when the mechanism is in place to try and control the mechanism by having its own nodes selected for the mechanism excluding an honest node. If this happens, the attack can be successful and it is possible a dishonest chain is confirmed. Thus, in contrast to a protocol operating under the LCR, with a Solomonic mechanism an attacker has an incentive to stake all of its nodes on both chains.

Suppose that an attacker has a share, $s$, of all nodes. If $A$ is the main chain and $B$ is the dishonest fork, as all nodes -- honest or not -- stake on each chain, the attacker may `control the mechanism' with probability, $s^2$. This assumes that the attacker does not know when engaging with the mechanism in the first round whether the other node is one of their own. This can be achieved by sequentially drawing those nodes with the first node forced to commit to a message before the second node is identified.

What does this do to the incentive to engage in a double-spend attack? Suppose that the fork is $k$ blocks long. Recall that $B$, as it was produced by the attacker, will result in block rewards of $kR$ being awarded to them if $B$ persists. If $A$ persists then that agent receives $skR$ in expected block rewards. Thus, the expected return to malicious agent is:
$$(1-s)skR+s\Big(-F+s(kR+\theta)+(1-s)skR\Big)$$
Note that, if the attacker does not attack, there is no fork and the attacker earns $skR$. Comparing this to the return to an attack, we can see that an attack will only take place if:
$$s > \frac{1}{2}+ \frac{\theta -\sqrt{(\theta +k R)^2-4 F k R}}{2 k R}$$
so long as $\frac{(k R + \theta)^2}{4kR} \ge F$. The higher is $F$, the less likely these conditions are satisfied and the less likely an attack occurs. Note that so long as $\theta \ge \frac{4F-kR}{2}$, the threshold to deter an attack is higher than the usual condition for LCR consensus that $s < \frac{1}{2}$. Thus, as with our mechanism for BFT consensus, here a mechanism has the potential to provide more security depending upon the choices of $F$ and $R$.

\subsection{Dealing with a Non-Traditional Double Spend Attack}

The above mechanism involves eliminating forks created by traditional double-spend (or double-spend like) attacks. The traditional double spend attack involves an attacker confirming a transaction that spends their tokens to the original chain and then removing that transaction with a dishonest chain in the hope of getting that latter chain accepted as consensus. This allows the attacker to spend those tokens again. The current Solomonic mechanism subverts that attack by making it impossible for the attacker to remove the original confirmed transaction as part of a forked chain.

The mechanism relies on the notion that attacks involve transactions being removed rather than retained. Thus, honest nodes have an incentive to defend transactions that are confirmed and dispute transactions that are proposed to be removed. The latter is also consistent with the effect of latency issues that may cause a fork when some nodes `miss' a transaction.

However, what if the attacker modifies their attack? Recall that the attacker wants to nullify a past transaction involving tokens spent by the attacker. While forking the chain is a potential way to do this, there are others. What if, for instance, the past transaction can be held to be invalid?

Normally, by virtue of being confirmed, the past transaction is valid -- that is, there are tokens in the wallet that can be sent to another wallet at that time. But if the attacker is forking a chain, the attacker, under POS, can fork the chain from an earlier block. In that block, the attacker can insert a transaction that spends the tokens (aka moves them to another wallet perhaps under the attacker's control) prior to the past transaction the attacker is trying to remove. If that is done as part of a fork, when the Solomonic mechanism is run, there will be two disputed transactions: one that appears on the dishonest chain but not on the original chain (the pre-spend) and one that appears on the original chain but the dishonest one (the original spend). When there are multiple disputed transactions, a natural way to run the mechanism would be to run it on the earliest disputed transaction first. The reason for this is that how that is resolved may have implications for future transactions. In fact, that is precisely what the attacker here is anticipating. Thus, the mechanism would, given the way it defaults to preserving transactions, preserve the pre-spend transaction and, by virtue of that, make the original spend transaction invalid. In other words, this attack can now succeed in generating the conditions of a double-spending of the tokens in question.

The problem here is that the current mechanism relies on (1) that honest nodes know they are honest and (2) that the only way an attack would work would be to nullify a transaction and honest nodes know that. However, (2) is overcome in the proposed attack and thus, the mechanism fails. In order to restore a mechanism here (2) would have to be replaced by conditions that the honest knew that the original chain was the `true' chain. 

What would those conditions be? Another way of looking at the fork chains is not transaction by transaction (to see which are disputed) but by paths of tokens. In the original chain, tokens may be confirmed to move from Charles to Alice while in the new chain tokens move from Charles to Bob. Thus, the dispute is over the end point of the tokens (or more generally, the path they have taken). In the case here, who owns the tokens now, Alice or Bob?

To resolve this, we propose using a variant of the Solomonic mechanism. The first step involves analysing the forked chains but instead of focusing on transactions that appear in one chain but not the other, the focus is on differences in the allocation of tokens in the last confirmed block of each chain (where we assume here that both chains are of equal length). If there has been a general double-spend attack, tokens that, say, were originally confirmed to be held by Alice would instead by held by Bob. 

The mechanism is as follows: If a fork appears (without loss of generality, $B$) and is within $x$ of the same number of blocks as fork $A$ then the following mechanism is run between nodes that claim to hold a full record of the blockchain.

\begin{enumerate}
    \item For each fork, the allocations of tokens confirmed on the last block are compared. 
    \item Transactions that lead to allocations that are in both sets are immediately confirmed at their original block time stamp. Other allocations are marked as disputed.
    \item One node from each chain is selected at random ($a$ for $A$ and $b$ for $B$). Each node is asked to confirm the current allocation of their chain based on their held full blockchain record. If they each agree on one chain, that chain is confirmed and the mechanism ends. If there is a disagreement, both are fined $F$, and they enter the dispute stage. 
\end{enumerate}

\noindent The \textbf{dispute stage} involves:
\begin{enumerate}
    \item One of the two nodes is selected at random and given the opportunity to assert the validity of their chain. 
    \item If they do not assert the chain, the tokens are confirmed to the account on the other chain.
    \item If they assert the chain, the tokens are burned.  
\end{enumerate}

We will assume that an attacker weakly prefers leaving tokens in the account on the original chain than burning those tokens.\footnote{Even with a double spend attack, the attacker faces some legal risk associated with being identified as the attacker which may happen if the tokens are missing from their original spend.}  Given this, we can prove the following.

\begin{proposition}
In any subgame perfect equilibrium, so long as at least one of the nodes participating in the mechanism is honest, a dishonest chain is never confirmed.
\end{proposition}

\begin{proof}
Without loss in generality let $A$ be the true fork and $B$ the dishonest fork. Working backwards, consider the dispute stage. There are four cases to consider:
\begin{enumerate}
    \item If $a$ (the node selected from $A$) is honest, if that agent asserts that $A$ is the true chain, this will result in the tokens being burned. $a$ receives $ - F$ as a payoff. If $a$ chooses not to assert that $A$ is true, $a$ receives $D-F$ which is a lower amount. Thus, $a$ asserts $A$ is true.
    \item If $a$ (the node selected from $A$) is dishonest, that agent will assert that $B$ is the true chain. This will result $B$ being confirmed which is what a dishonest node wants. 
    \item If $b$ (the node selected from $B$) is honest, that agent will assert that $A$ is the true chain. This will result in $A$ being confirmed which is what an honest node wants.
    \item If $b$ (the node selected from $B$) is dishonest, if that agent asserts that $B$ is the true chain, this will result tokens being burned. If they assert that $A$ is the true chain, this will result in $A$ being confirmed.
\end{enumerate}
Now consider the choice of whether of which chain to claim in the first round. An honest node will always find it worthwhile to assert that $A$ is the true chain. A dishonest node, under the assumptions of the Proposition, knows the other node is an honest node. If that node asserts that $A$ is true, the original chain is confirmed. If that node asserts that $B$ is true, the dispute chain begins. This will result in either the tokens remaining on $A$ or the tokens being burned. So long as the former is weakly preferred by dishonest nodes, that node will claim $A$ is the true chain in order to avoid the fine, $F$. This proves the proposition.
\end{proof}
\\

\noindent As before, the presence of an honest node, prevents the attacker from achieving their aims to have the tokens allocated as in $B$ rather than $A$. If the attack involves any cost, the attack will, therefore, not take place.

That said, this presumes that an honest node is part of the mechanism. This may not happen and the attacker may be able to control both nodes in the mechanism. Put differently this mechanism relies on honest nodes being able to tell which is the correct chain and hence, who rightfully owns the tokens.

This is an assumption we make here, and there may be practical difficulties in nodes being able to have a record of the blockchain and which chain was the main chain in the immediate past. That said, there is good reason to think that this assumption is justified. Honest nodes have a lot of information when they see a fork going back many transactions in the chain.

While there will always be forks due to latency issues, there will never be a ``dishonest fork'' on the equilibrium path of the game induced by our mechanism. Indeed, provided that the honest nodes have a higher prior that the correct chain is, indeed, the correct chain -- that is, priors are biased towards the truth -- the mechanism here will still have them acting as if the correct chain is the truth and will still deter attempts to confirm a dishonest chain as the consensus chain.

\section{Concluding Remarks}

We have shown how to construct revelation mechanisms to achieve consensus on blockchains under BFT and LCR. A fundamental pillar of a mechanism-design approach to blockchain consensus is the use of information. Consistent with the central dogma of mechanism design, the designer is not presumed to posses more information than that held by agents in the mechanism. This contrasts, however, with some uses of slashing in existing approaches to consensus. Yet, through careful design choices, the mechanisms we examine here make very efficient of the information that is held by existing nodes. We have also discussed the robustness of these mechanisms to multi-node attacks.

We are of the view that mechanism-designed-based consensus protocols have important advantages over existing POS protocols, and that they are likely to be of practical use as an alternative to POW protocols. In particular, rather than just being designed to satisfy stringent requirements of finality and liveness (themselves often probabilistic), mechanisms have choice parameters (e.g., the reward, $R$, and fines, $F$) that themselves can manipulate trade-offs between finality and liveness at the margin depending on the environment and preferences of blockchain users. This expands the set of blockchain consensus options for participants.

\newpage

\typeout{}
\bibliography{references}
\bibliographystyle{apalike}

\end{document}